\documentclass[12pt,notitlepage]{iopart}

\usepackage[latin1]{inputenc}
\usepackage{subfigure}
\usepackage{graphicx}
\usepackage{dsfont}

\usepackage{rotating}

\newcommand{\R}{\mathds{R}}
\newcommand{\N}{\mathds{N}}
\newcommand{\1}{\mathds{1}}

\newcommand{\ket}[1]{| #1 \rangle}

\usepackage{amsthm}
\newtheorem{theorem}{Theorem}
\newtheorem{corollary}{Corollary}

\newtheorem*{theorem*}{Theorem}
\theoremstyle{definition}
\newtheorem{example}{Example}

\usepackage{hyperref}
\hypersetup{
    colorlinks,
    citecolor=black,
    filecolor=black,
    linkcolor=black,
    urlcolor=black
}

\makeatletter
\def\Ddots{\mathinner{\mkern1mu\raise\p@
\vbox{\kern7\p@\hbox{.}}\mkern2mu
\raise4\p@\hbox{.}\mkern2mu\raise7\p@\hbox{.}\mkern1mu}}
\makeatother

\begin{document}

\title{Optimisation of Bell inequalities with invariant Tsirelson bound}

\author{M Epping$^1$, H Kampermann$^{1}$ and D Bru\ss$^{1}$}
\address{$^1$ Institut f\"{u}r Theoretische Physik III, Heinrich-Heine-Universit\"{a}t D\"{u}sseldorf, Universit\"{a}tsstrasse 1, D-40225
D\"{u}sseldorf, Germany}
\ead{\mailto{epping@thphy.uni-duesseldorf.de}}
\date{\today}


\begin{abstract}
We consider a subclass of bipartite CHSH-type Bell inequalities. We investigate operations, which leave their Tsirelson bound
invariant, but change their classical bound. The optimal observables are unaffected except for a relative rotation of the two laboratories.
We illustrate the utility of these operations by giving explicit
examples: We prove that for a fixed quantum state and fixed measurement setup except for a relative rotation of the two laboratories, there
is a Bell inequality that is maximally violated for this rotation, and we optimise some Bell inequalities with respect to the maximal
violation. Finally we optimise the qutrit to qubit ratio of
some dimension witnessing Bell inequalities.
\end{abstract}
{\let\newpage\relax\maketitle}
\vfill
\tableofcontents
\vfill
\clearpage

\markboth{Optimisation of Bell inequalities with invariant Tsirelson bound}{}

\section{Introduction}\label{sec:introduction}
Originally, John S. Bell introduced what we now call Bell inequalities in order to show that the ideas of locality and realism are
incompatible with statistical predictions of quantum theory~\cite{Bell64}. Thus, the question whether a completion of quantum theory obeying
these axioms exists, as proposed by Einstein, Podolski and Rosen~\cite{EPR35}, was brought to an experimentally testable level. Now,
fifty years later, there is very strong experimental evidence that Bell inequalities can be violated by
nature~\cite{Aspect82,Zeilinger98,Wineland01,Ansmann09,Christensen2013,Shadbolt2012,2013arXiv1309.1379E,2013arXiv1312.4810L}, which implies that not
all axioms in the derivation of Bell inequalities are followed by nature. Nevertheless Bell inequalities are not water under the bridge
yet. This is amongst other reasons due to several interesting applications, like quantum key distribution, where the
violation of Bell
inequalities is a test for eavesdropping~\cite{Ekert91}. Here and in other applications, the amount of violation becomes important and a stronger
violation of the inequality is usually beneficial (e.g. noise is less corruptive or the gap between classical and quantum performance
increases).\\
In the present paper, we discuss two methods to modify Bell inequalities, which change the classical bound but leave the maximal value
achievable
in quantum theory unchanged. These methods can be used to optimise Bell inequalities with respect to the possible amount of violation.
Various research on Bell inequalities with a large amount of violation has been carried
out~\cite{2010arXiv1012.5043B,Junge2011,PhysRevLett.103.180402,PhysRevA.82.042334}, but literature on the specific problem investigated in
this paper is less extensive~\cite{Fishburn1994,PhysRevA.77.032108}.\\
We specify the Bell inequalities under consideration in the following
Section~\ref{sec:subclass}. Then we formulate the above-mentioned methods as a
Corollary in Section~\ref{sec:methods} and give examples for their utility in Section~\ref{sec:applications}. Section~\ref{sec:conclusions}
concludes this paper.
\section{A subclass of CHSH-type Bell inequalities}\label{sec:subclass}
We consider bipartite full correlation Bell inequalities (CHSH type Bell inequalities~\cite{Clauser1969,WernerWolf01}) with $M_i$
measurement settings at the site of party $i$. These
settings are labelled $x_i=1,2,...,M_i$. Such Bell inequalities can be written in the form
\begin{equation}
 \sum_{x_1,x_2=1}^{M_1,M_2} g_{x_1,x_2} E(x_1,x_2) \leq B, \label{eq:BIwithoutmax}
\end{equation}
where $E(x_1,x_2)$ is the expectation value for setting $x_1$ at Alice's site and $x_2$ at Bob's site and $g$ is a real $M_1\times
M_2$-matrix of coefficients. Measurement outcomes are required to be in the interval $[-1,1]$. The dimension of the two subsystems is not
fixed. The local hidden variable bound $B$ holds
for all values
achievable in local hidden variable theories, i.e.
\begin{equation}
 \max_{a_1,a_2} \sum_{x_1,x_2=1}^{M_1,M_2} g_{x_1,x_2} a_1(x_1) a_2(x_2)\leq B. \label{eq:bellineq}
\end{equation}
$B$ can be calculated by performing this maximisation over all possible (deterministically) predefined measurement outcomes $a_1(x_1)=\pm 1$
and $a_2(x_2)=\pm 1$. Due to the assumption of locality, $a_1$ ($a_2$) does not depend on $x_2$ ($x_1$). The use of unmeasured outcomes is
motivated by the assumption of realism. See \cite{Peres86} for a more thorough analysis. For some $g$, Inequality~(\ref{eq:bellineq}) can be
violated
within quantum theory.\\ 
Similarly, one can write down bounds for expectation values predicted by quantum theory~\cite{Tsirelson80}. The analogue of
Ineq~(\ref{eq:BIwithoutmax}) reads
\begin{equation}
 \sum_{x_1,x_2=1}^{M_1,M_2} g_{x_1,x_2} E(x_1,x_2) \leq T, \label{eq:tsirelsonineq}
\end{equation}
where $T$ is a Tsirelson bound, which holds for all quantum states given by a density matrix $\rho$ and all observables
$\mathcal{A}_1(x_1)$ and $\mathcal{A}_2(x_2)$, i.e.
\begin{equation}
 \max_{\mathcal{A}_1,\mathcal{A}_2,\rho} \sum_{x_1,x_2=1}^{M_1,M_2} g_{x_1,x_2} \tr\left(\rho \mathcal{A}_1(x_1)\otimes
\mathcal{A}_2(x_2)\right)\leq T.  \label{eq:tsirelsonbound}
\end{equation}
Remember that we did not restrict the dimension of the Hilbert space. In \cite{Epping2013} we showed that a quantum bound for
Inequality~(\ref{eq:tsirelsonineq}) is given by
\begin{equation}
 T(g)=||g||_2 \sqrt{M_1 M_2}, \label{eq:singularvaluebound}
\end{equation}
where $||g||_2$ is the largest singular value of $g$. However, this bound $T$ is not always tight: It is not always possible to
achieve equality in Inequality~(\ref{eq:tsirelsonbound}). Nevertheless it is tight for a subclass of Bell inequalities,
which contains many well-known Bell inequalities. In this paper we will restrict ourselves to this class of Bell inequalities, for which
$T$ in Eq.~(\ref{eq:singularvaluebound}) is achievable for some states and observables. In this case the violation of the Bell inequality,
which is the ratio of the quantum and the classical value, is
\begin{equation}
 \nu = \frac{T}{B} \label{eq:nu}.
\end{equation}
According to a theorem by Tsirelson~\cite{Tsirelson1993}, there exist real vectors $\vec{v}_1$, $\vec{v}_2$, ..., $\vec{v}_{M_1}$
and $\vec{w}_1$,
$\vec{w}_2$, ..., $\vec{w}_{M_2}$, such that the quantum mechanical expectation value can be written as
\begin{equation}
 E(x_1,x_2)=\vec{v}_{x_1}^T \vec{w}_{x_2}. \label{eq:Eundvw}
\end{equation}
In the present context, it is usually more convenient to use these vectors instead of the observables. Let $g$ be a real $M_1\times
M_2$-matrix and $V$, $S$, $W$ be a singular value decomposition of $g$, i.e. $g=V S W^T$ with diagonal $S$ and
$V$, $W$ being orthogonal. We denote the dimension of the space of the largest singular value $||g||_2$ as $d$, i.e. this is the
degeneracy of the largest singular value. The corresponding matrices of the
truncated singular value decomposition associated with $||g||_2$ contain the first $d$ columns (the singular vectors) of $V$ and $W$,
respectively.
\begin{theorem}[Tightness of $T$~\cite{Epping2013}]\label{thm:tightness}
For any real $M_1\times M_2$ matrix $g$, let $V^d$, $||g||_2 \1^d$, $W^d$ be a truncated singular value decomposition of $g$ associated with
$||g||_2$, where $d$ is the degeneracy of $||g||_2$. The bound $T=||g||_2 \sqrt{M_1 M_2}$ can be reached with
observables, which are
linked via $E(i,j)=\vec{v}_{i}^T\vec{w}_{j}$ to $d'\leq
d$-dimensional real
vectors $\vec{v}_i$ and $\vec{w}_j$ given by
\begin{eqnarray}
 \hphantom{\mbox{and }{}}\vec{v}_i &=& \alpha^T V^d_{i,*} \label{eq:obsv}\\
 \mbox{and }\vec{w}_j &=& \sqrt{\frac{M_2}{M_1}}\alpha^T W^d_{j,*},\label{eq:obsw}
\end{eqnarray}
if and only if there exists a $d\times d'$-matrix $\alpha$, such that these vectors are normalised. Here $V^d_{i,*}$ and $W^d_{j,*}$ denote
column vectors
containing the elements of the $i$-th row of $V^d$ and the $j$-th row of $W^d$, respectively. 
\end{theorem}
There is a geometric interpretation of the norm conditions: The bound $T$ is tight for observables corresponding to $d'$-dimensional real
vectors $\vec{v}_i$ and $\vec{w}_j$ if and only if the vectors $V_{i,*}^d$ and $\sqrt{\frac{M_2}{M_1}} W_{j,*}^d$ lie on the surface of an
origin-centred ellipsoid with no more than $d'$ finite semi-axes~\cite{Epping2013}. We call this object a $d'$-dimensional ellipsoid.

\section{Modifying Bell inequalities inside this class}\label{sec:methods}
We aim at modifying Bell inequalities inside the class described in the previous section, i.e. those where the quantum bound given in Eq.~(\ref{eq:singularvaluebound}) is tight. In particular, we are interested in operations
that do not change the value of $T$ given in Eq.~(\ref{eq:singularvaluebound}). However, in general these operations do change the classical
bound of the Bell inequality, i.e. the modification's effect on the quantum and the classical value are qualitatively and quantitatively
different. This is in contrast to arbitrary modifications of the coefficients, where both values are simultaneously affected.
We will exemplify later that such modifications can be a useful tool, e.g. for optimising Bell inequalities. The following corollary gives
modifications with the properties we are seeking.

\begin{corollary}\label{cor:mods}
Let $g$ be a $M_1\times M_2$ real matrix with singular value decomposition $V$,$S$,$W$, i.e. $g=V S W^T$, such that $T(g)$ is achievable.
The multiplicity of $||g||_2$ is denoted by $d$, the length of the diagonal of $S$ is $s=\min(M_1,M_2)$. 
The following modifications of $g$ lead to achievable bounds $T(g')$ (primed symbols correspond to the modified coefficients $g'$).
\begin{enumerate}
\item\label{mods:twist} "Twisting" of singular vectors:\\
 For
\begin{equation}
  g' =V \left(\begin{array}{cc}
                                           R_1 & 0 \\ 0 & R_2
                                          \end{array}\right) S \left(\begin{array}{cc}
                                           \1^d & 0 \\ 0 & R_3
                                          \end{array}\right) W^T,
\end{equation}
where $R_1$ is a $d\times d$ orthogonal matrix commuting with $\alpha$ (see Eq.~(\ref{eq:obsv})) and $R_2$ and $R_3$ are orthogonal matrices
of dimension $(M_1-d)$ and
$(M_2-d)$, respectively, $T(g')=T(g)$ is achievable.
\item\label{mods:shift} Modification of singular values:\\
 For signs $\sigma_1,...,\sigma_d=\pm 1$ and real numbers $\lambda_1,\lambda_{d+1},\lambda_{d+2},...,\lambda_s$ fulfilling
\begin{equation}
|\lambda_i+S_{i,i}| < \left|\,||g||_2+\lambda_1\right| \mbox{ for all }i>d \label{eq:diagonalcond}
\end{equation} the modified coefficients $g'=V S' W^T$ with 
\begin{equation}
S' =
\mbox{diag}(\sigma_1(S_{1,1}+\lambda_1),...,\sigma_d(S_{d,d}+\lambda_1),S_{d+1,d+1}+\lambda_{d+1},...,S_{s,s}+\lambda_s)
\end{equation}
correspond to an inequality with achievable $T(g')=\frac{||g||_2+\lambda_1}{||g||_2} T(g)$.
\end{enumerate}
\end{corollary}
\begin{proof}
\begin{enumerate}
 \item $R_1$ can be considered as a rotation of the singular vectors, i.e. the singular values in $S$ are not affected. If $R_1$ and the
$d\times d$-matrix $\alpha$ commute, then
\begin{equation}
 ||\alpha^T (R_1^T V_{i,*}^d)||=||R_1^T \alpha^T  V_{i,*}^d||=||\alpha^T  V_{i,*}^d||,
\end{equation}
i.e. the conditions of Theorem~\ref{thm:tightness} are not affected. $R_2$ and $R_3$ merely rotate the singular vectors outside the space
associated with $||g||_2$, which neither affects the tightness nor the value of the bound.
\item The conditions for tightness according to Theorem~\ref{thm:tightness} and the value of $T$ are not affected by modification of
non-maximal singular values, as long as they do not become maximal. Adding the same value to all largest singular values is only a scaling
of $T$ (as long as they remain maximal). A negative diagonal entry induces a sign change of the elements of the corresponding singular
vector.
\end{enumerate}
\end{proof}
\noindent Please note that the condition of (\ref{mods:twist}) is fulfilled, if there exists a solution $\alpha\propto \1$.\\
We remark that (\ref{mods:shift}) is a generalisation of the diagonal modification in \cite{Fishburn1994}. There, $V=W$ and
$g'=g+\lambda \1$, which corresponds to $\lambda_i=\lambda$. The condition of Eq.~(\ref{eq:diagonalcond}) on the $\lambda_i$ can be ignored,
if one assures tightness according to Theorem~\ref{thm:tightness}. For example, depending on the particular form of $V$ and $W$, the bound
might be tight for different values of $d$. In particular, $T$ is tight, if the new singular values are all equal. Furthermore
(\ref{mods:shift}) includes the special case, where $g'=r g$ for any $r\in\R$.
\section{Using these modifications as a method}\label{sec:applications}
In this section, the modifications described in Corollary~\ref{cor:mods} are applied to specific examples of
coefficient matrices $g$.
\subsection{Maximally violated Bell inequality for relative rotation of laboratories}
We start with the modification (\ref{mods:twist}), i.e. the twisting of the singular vectors. First we note that
$R_1$ is a relative rotation of the two
parties in the sense, that the real vectors $\vec{v}_{x_1}$, which define optimal observables for party one, are rotated by $R_1$. For
$d'=3$, one can interpret $\vec{v}_{x_1}$ as a Bloch-vector, i.e. $\mathcal{A}(x_1)=\vec{v}_{x_1}^T \vec{\sigma}$, where $\vec{\sigma}$ denotes the vector
containing the three Pauli matrices. In this way we see that the rotation $R_1$ corresponds to a relative rotation of the two laboratories
in the
usual sense.\\
This motivates us to prove the following statement.
\begin{example}
For every relative rotation between the laboratories of party one and party two, there exists a
Bell inequality that is maximally violated for exactly this rotation. 
We require that the experimental setup is fixed up to the relative rotation, i.e. the measurement directions in the local coordinate
systems and the shared state (e.g. $\ket{\phi^+}=\frac{1}{\sqrt{2}}(\ket{00}+\ket{11})$) do not depend on the rotation angle. Consider the
Bell inequalities given via the coefficient matrix
\begin{equation}
g=\left(
\begin{array}{ccc}
 \frac{1}{2} & \frac{1}{2} & 0 \\
 -\frac{1}{2} & \frac{1}{2} & 0 \\
 0 & 0 & \frac{1}{\sqrt{2}} \\
 \frac{1}{\sqrt{2}} & 0 & 0 \\
 0 & \frac{1}{\sqrt{2}} & 0 \\
 0 & 0 & \frac{1}{\sqrt{2}} \\
\end{array}
\right) R(\Phi,\Theta,\Psi), \label{eq:gedrehteBI}
\end{equation}
where
\begin{eqnarray}
R(\Phi,\Theta,\Psi)&=&\phantom{{}\times{}}
\left(
\begin{array}{ccc}
 1 & 0 & 0 \\
 0 & \cos (\Phi ) & \sin (\Phi ) \\
 0 & -\sin (\Phi ) & \cos (\Phi ) \\
\end{array}
\right)\times \nonumber\\
&&\times
\left(
\begin{array}{ccc}
 \cos (\Theta ) & 0 & -\sin (\Theta ) \\
 0 & 1 & 0 \\
 \sin (\Theta ) & 0 & \cos (\Theta ) \\
\end{array}
\right)
\left(
\begin{array}{ccc}
 \cos (\Psi ) & \sin (\Psi ) & 0 \\
 -\sin (\Psi ) & \cos (\Psi ) & 0 \\
 0 & 0 & 1 \\
\end{array}
\right)
\end{eqnarray}
is a general rotation given by the roll-pitch-yaw angle. From Eq~(\ref{eq:gedrehteBI}) one can read the truncated singular value
decomposition associated with the threefold ($d=3$) degenerate maximal singular value $1$: i.e the first factor is $V^d$, $S^d=\1$ and
$W^T=R(\Phi,\Theta,\Psi)$. From this we already know that
$T(g)=\sqrt{M_1 M_2}=3\sqrt{2}$ for all angles. This bound is achievable, because $\alpha=\sqrt{2}\1^3$ is a
solution (see Theorem~\ref{thm:tightness}). The vectors $V_{1,*}^d$, $V_{2,*}^d$, $V_{4,*}^d$, $V_{5,*}^d$ force the rank of $\alpha$ to be
at least two (two or more semi-axes of the corresponding ellipsoid are finite). Therefore, the inequality associated with $g$ is really a
Bell inequality,
i.e. it can be violated. The bound is achieved for the state $\ket{\phi^+}=\frac{1}{\sqrt{2}}(\ket{00}+\ket{11})$ with
observables
\begin{eqnarray}
\mathcal{A}(x_1)&=&\vec{v}_{x_1}^T \vec{\sigma},\\
\mathcal{A}(x_2)&=&(\vec{w}_{x_2}^T \vec{\sigma})^T.
\end{eqnarray}
Because
\begin{equation}
\vec{w}_j = \sqrt{\frac{M_2}{M_1}} \alpha^T W_{j,*}^d = R(\Phi,\Theta,\Psi)_{j,*},
\end{equation}
i.e. the measurement directions of party two are given by the columns of $R^T(\Phi,\Theta,\Psi)$, this Bell inequality is maximally violated
for a relative rotation of the laboratories given by the roll-pitch-yaw angle (see Figure~\ref{fig:rotatedbellineq} (a)). The violation
$\frac{T}{B}$ of the inequality given by the coefficients in Eq.~(\ref{eq:gedrehteBI}) depends on the angles (see
Figure~\ref{fig:rotatedbellineq} (b,c,d)).
\end{example}
\begin{figure}
 \subfigure[ ]{ \includegraphics[height=0.245 \linewidth]{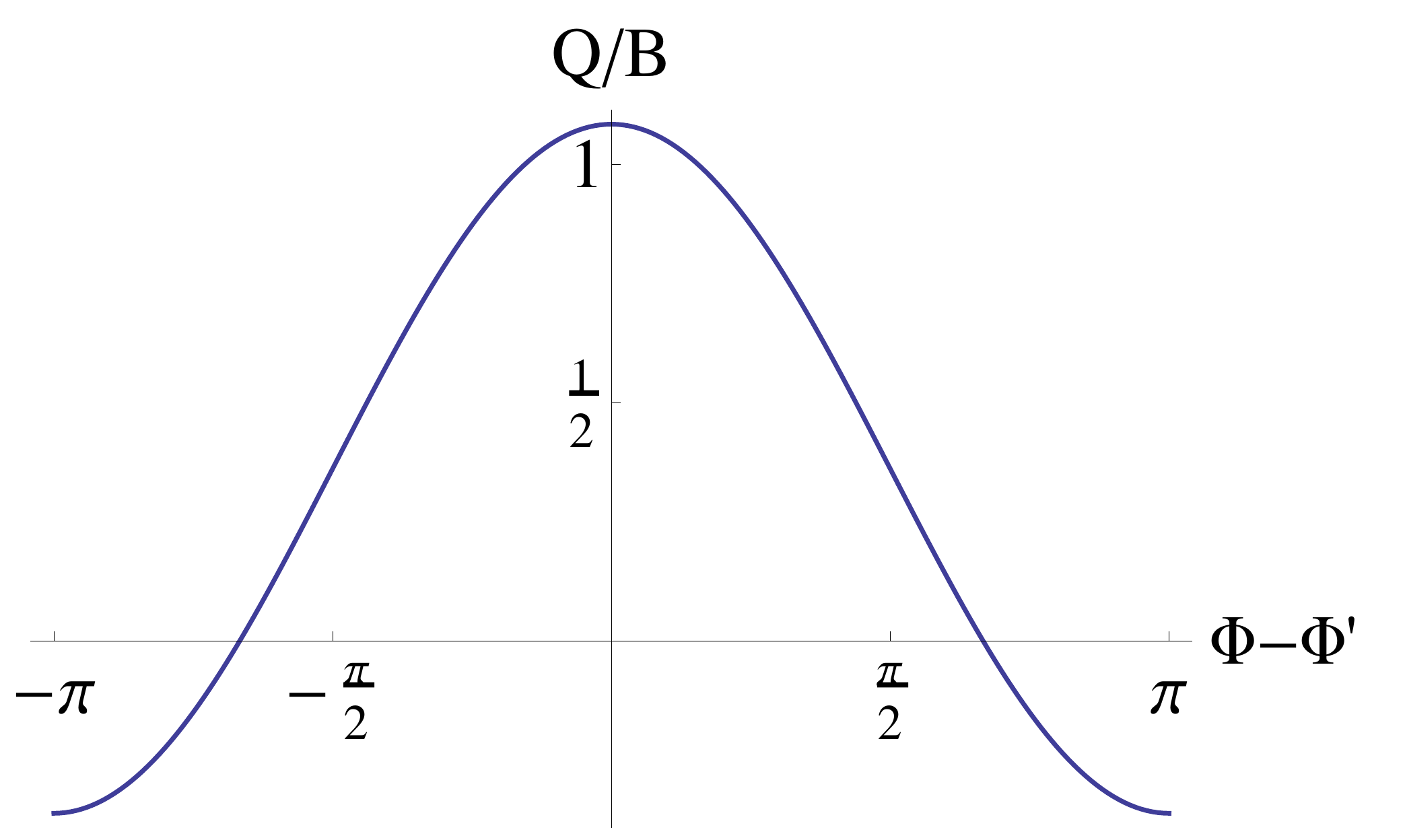} }\\
 \subfigure[$\Psi=\frac{1}{4}\pi$]{ \includegraphics[height=0.245 \linewidth]{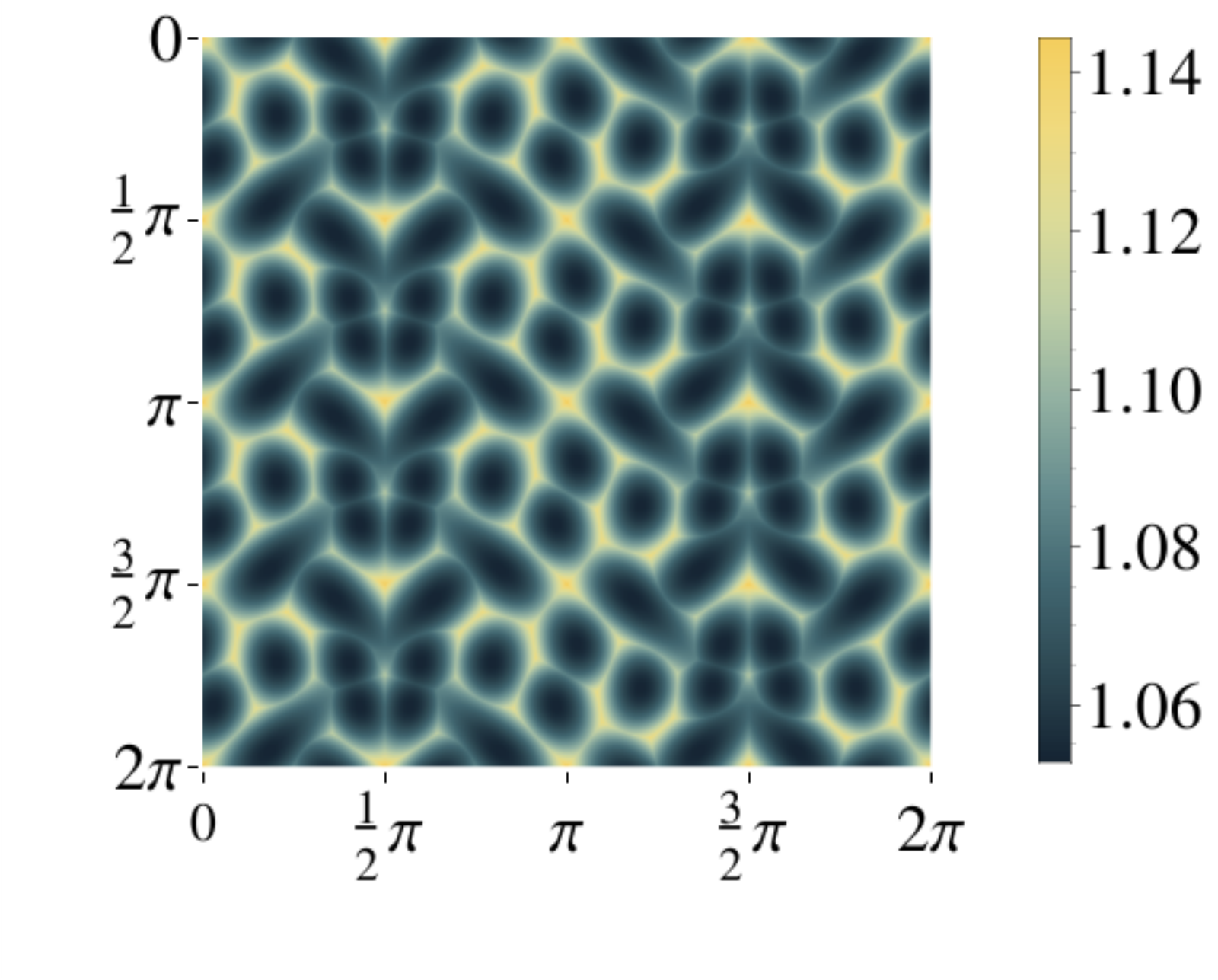} } \hfill
 \subfigure[$\Psi=\frac{1}{2}\pi$]{ \includegraphics[height=0.245 \linewidth]{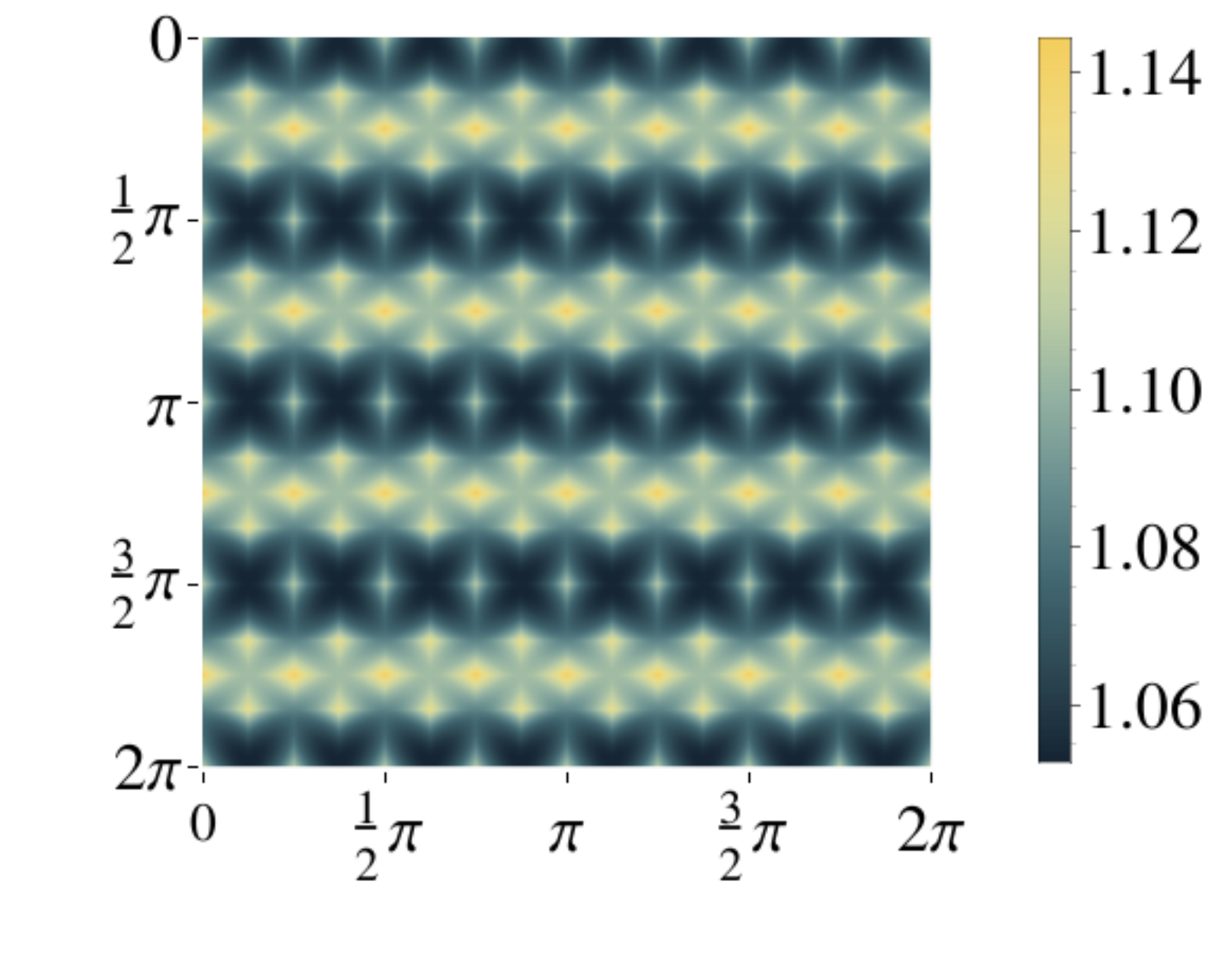} } \hfill
 \subfigure[$\Psi=\pi$]{ \includegraphics[height=0.245 \linewidth]{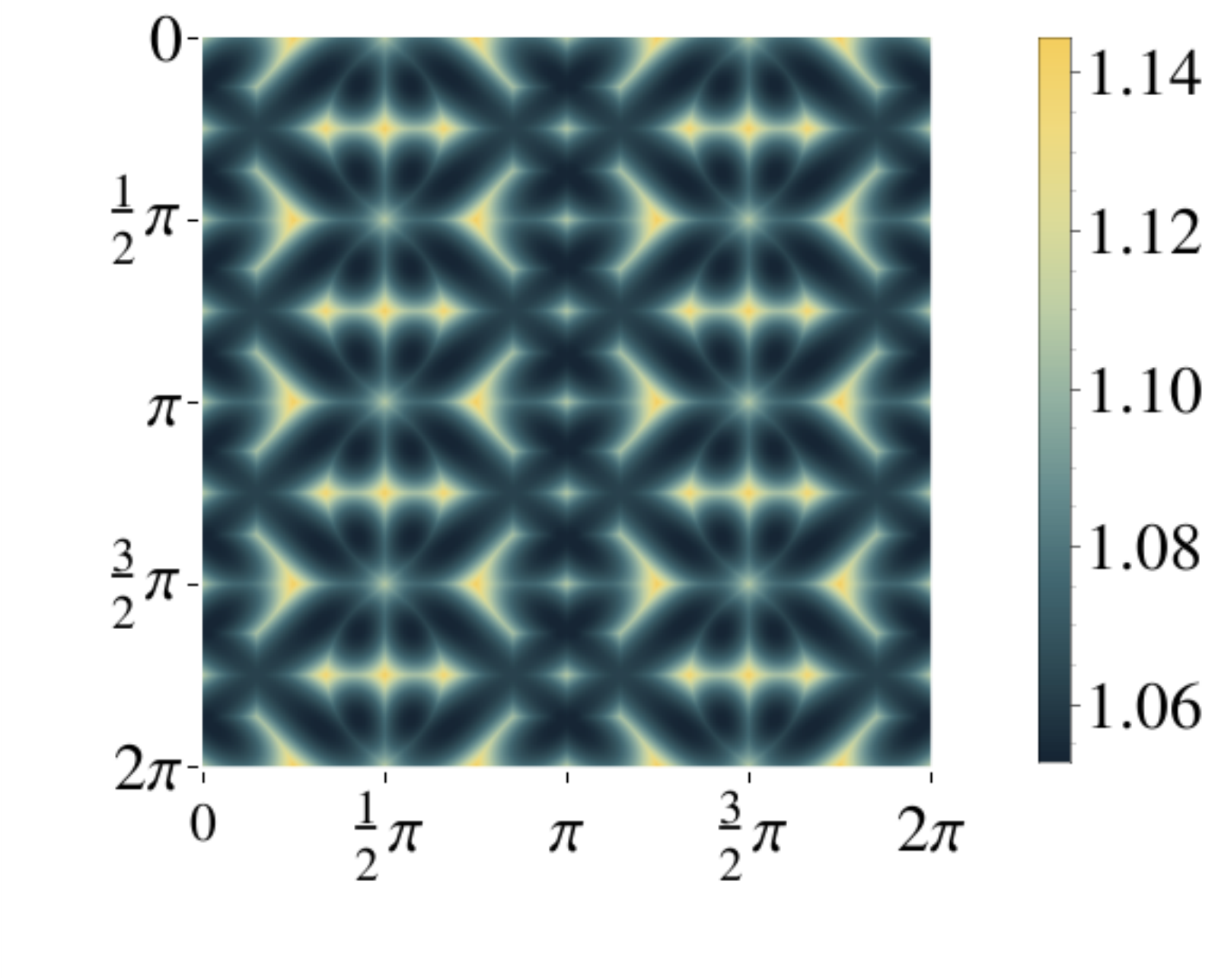} }
 \caption{Violation of the ``rotated'' Bell inequality (Eq.~(\ref{eq:gedrehteBI})). (a) The measured violation (quantum value $Q$ for actual
observables devided by local hidden variable value $B$) of the Bell inequality with
coefficients given in Eq~(\ref{eq:gedrehteBI}) for optimal angles $(\Phi,\Theta,\Psi)$ depending on the yaw angle $\Phi'$ of the actual
observables. $\Phi$, $\Theta$ and $\Psi$ are fixed to arbitrary
values. The same plot can be drawn for $\Theta$ and $\Psi$. (b,c,d) The maximal violation of Bell inequalities given by different angles,
see Eq.~(\ref{eq:gedrehteBI}), where $\Psi$ is fixed.}
 \label{fig:rotatedbellineq}
\end{figure}

\subsection{Optimisation of Bell inequalities for fixed measurement directions}
In several applications a large violation is desirable. Given the experimental measurement setup used to evaluate a given Bell inequality, there might
be different inequalities that lead to a higher violation. In that sense, they are "better" inequalities. Finding an optimal inequality
seems to be a difficult task. In some cases the methods above (Corollary~\ref{cor:mods} (\ref{mods:twist}) and (\ref{mods:shift})) might
give an intuition how to improve a given matrix of coefficients $g$ without changing the involved measurements.\\
There is another possible motivation for restricting the observables in the optimisation of the violation: It turns out that the average
violation of Bell inequalities inside this restricted parameter space is larger than the one for the whole parameter space.
Figure~\ref{fig:histogram} shows the probability of an amount of violation for completely random coefficients and rotated versions
(Corollary~\ref{cor:mods} (\ref{mods:twist})) of
\begin{equation}
 g=\left(
\begin{array}{ccc}
 1 & -1 & -1 \\
 1 & 1 & -1 \\
 1 & 1 & 1 \\
\end{array}
\right).
\end{equation}

\begin{figure}[tbp]
 \begin{center}
\includegraphics[height=0.245\textwidth]{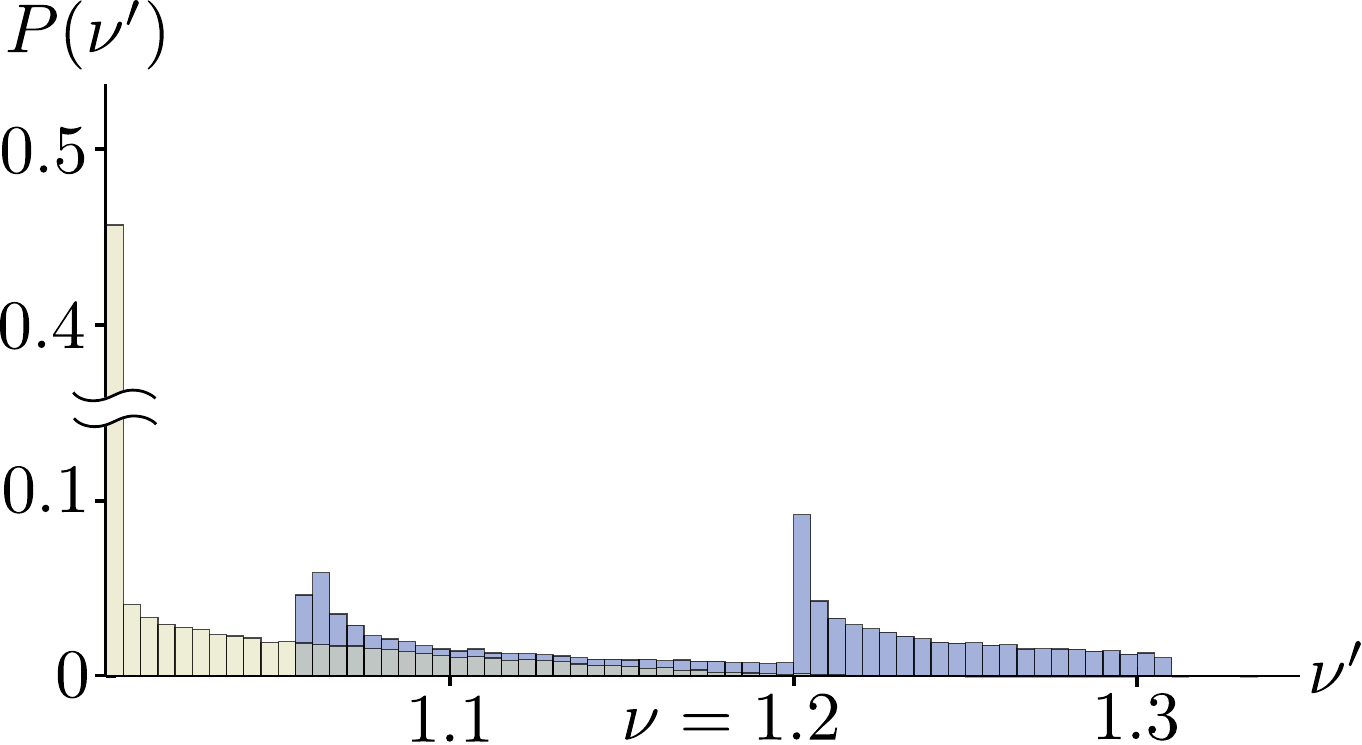}
  
 \end{center}

 \caption{A histogram of the maximal violation $\nu'$ of a random Bell inequality (light) and a ``twisted'' version of Gisin's
inequality~\cite{Gisin1999}
(dark). The size of the matrices is $3\times 3$. The random inequality has equally distributed coefficients in $[-1,1]$. For the other
inequality, the rotation angles are equally distributed. The probability of a given violation is estimated from samples of $50000$
inequalities each. The violation of the original inequality by Gisin is  $\nu=1.2$.}
 \label{fig:histogram}
\end{figure}

\begin{example}
The coefficients of Gisin's inequality~\cite{Gisin1999} for $M_1=M_2=6$ read 
\begin{equation}
g=\left(
\begin{array}{cccccc}
 1 & -1 & -1 & -1 & -1 & -1 \\
 1 & 1 & -1 & -1 & -1 & -1 \\
 1 & 1 & 1 & -1 & -1 & -1 \\
 1 & 1 & 1 & 1 & -1 & -1 \\
 1 & 1 & 1 & 1 & 1 & -1 \\
 1 & 1 & 1 & 1 & 1 & 1 \\
\end{array}
\right).
\end{equation}
This inequality has $B=18$, as one can easily see when the first two rows get multiplied with $-1$. The quantum value is
$T=M/\sin(\pi/(2M))=12 \sqrt{2+\sqrt{3}}\approx 23.1822$. Using Corollary~\ref{cor:mods} (\ref{mods:shift}) we can optimise the coefficients
numerically, and obtain
\begin{eqnarray}
g'&=&V \mbox{diag }(||g||_2,||g||_2,||g||_2,||g||_2,-||g||_2,-||g||_2) W^T \nonumber\\
 &=& (1+\sqrt{3})\left(
\begin{array}{cccccc}
 0 & 0 & -1 & 0 & 0 & -1 \\
 1 & 0 & 0 & -1 & 0 & 0 \\
 0 & 1 & 0 & 0 & -1 & 0 \\
 0 & 0 & 1 & 0 & 0 & -1 \\
 1 & 0 & 0 & 1 & 0 & 0 \\
 0 & 1 & 0 & 0 & 1 & 0 \\
\end{array}
\right),
\end{eqnarray}
which is equivalent to the CHSH inequality. This implies a violation of $\nu'=\sqrt{2}$ and $B(g')=6(1+\sqrt{3})\approx 16.3923$. Here we
ignored
the condition in Eq~(\ref{eq:diagonalcond}) of Corollary~\ref{cor:mods}~(\ref{mods:shift}) as tightness of $T$ is ensured by the fact that
all singular values are equal. One would obtain the same result, when considering $g'=V \mbox{diag
}(||g||_2,||g||_2,||g||_2-\varepsilon,||g||_2-\varepsilon,-||g||_2+\varepsilon,-||g||_2+\varepsilon) W^T $ for a very small positive
$\varepsilon$. In this way the degeneracy remains $d'=2$ and the condition of Eq.~(\ref{eq:diagonalcond}) is fulfilled. The matrix $g'$
constitutes a local optimum, i.e. small modifications of the singular values lead to a smaller violation.
\end{example}
\begin{example}[Fishburn-Reeds Inequalities~\cite{Fishburn1994}]
In \cite{Fishburn1994}, the authors construct a series of inequalities with increasing number of measurement settings. For $d\in\N$ greater
or equal two,
\begin{equation}
 g=V^d (V^d)^T-\frac{4}{3} \1,
\end{equation}
where $V^d$ is a $(d-1)d\times d$ matrix containing all rows of the form $(-1,0,...,0,1,0,...,0)$ and $(1,0,...,0,1,0,...)$. The columns of
$V^d$ are orthogonal and thus $\frac{1}{\sqrt{2 (d-1)}} V^d$, $(2(d-1)-4/3)\1^d$ and $\frac{1}{\sqrt{2 (d-1)}} V^d$ form a truncated
singular value decomposition of $g$. Therefore, the optimal measurement settings for party one and party two are identical. Intuitively,
this choice of settings seems to be not optimal with respect to the amount of violation. We searched numerically for inequalities with a
larger violation using methods (\ref{mods:twist}) and (\ref{mods:shift}) of Corollary~\ref{cor:mods}. We give improved violations for
$d=2,...,5$ in Table~\ref{tab:frprime}. Due to the computational complexity of determining $B$, it is likely that the given values are not
the maximal ones achievable with these methods.
\end{example}
\begin{table}[bt]
\centering
\begin{tabular}{|c|c|c|c}
\hline
 $d$ & $\nu=T/B$ & $\nu'=T/B'$\\
\hline
$2$ & $1$ & $\sqrt{2}\approx 1.41421$\\
$3$ & $4/3\approx 1.33333$ & $1.34163$\\
$4$ & $7/5=1.4$ & $\sqrt{2} \approx 1.41421$\\
$5$ & $10/7\approx 1.42857$ & $1.42860$\\
\hline
\end{tabular} 
\caption{Optimized violations of the first four inequalities by Fishburn and Reeds~\cite{Fishburn1994}, $T/B'$, compared
to the original violation $T/B$. The explicit coefficients of the corresponding matrix $g'$ are given in the supplemental material. Note
that the given values for $\nu'$ are not necessarily maximal.}
\label{tab:frprime}
\end{table}

\subsection{Optimisation of dimension witnessing Bell inequalities}
The minimal $d'$ for a solution $\alpha$ is a lower bound on the length of the vectors $\vec{v}_i$ and $\vec{w}_j$, which is linked to the
dimension of the observables. For example, if this minimal $d'$ is larger than three, the maximal quantum value of the inequality cannot be
reached using qubits. Let us denote the bound for $d'$-dimensional real vectors by $T_{d'}$. Please note that $B=T_1$.\\
In the previous section, we aimed at increasing the ratio $T/T_1$ by decreasing $T_1$. The same optimisations can be performed for any other
value $d'$ with $T_{d'}<T$.\\
To calculate the bound $T_{d'}$, we are interested in the optimal strategy (optimal "observables") achieving this bound. We note that the
optimal observables of party two are fixed by the ones of party one. The maximum in Eq.~(\ref{eq:tsirelsonbound}) using
Eq.~(\ref{eq:Eundvw}) is achieved, if for all $x_2$ (each column), the vector $\vec{w}_{x_2}$ is parallel to $\sum_{x_1} g_{x_1,x_2}
\vec{v}_{x_1}$, i.e.
\begin{equation}
 \vec{w}_{x_2}=\frac{1}{||\sum g_{x_1,x_2} \vec{v}_{x_1}||} \sum g_{x_1,x_2} \vec{v}_{x_1},
\end{equation}
so the bound simplifies to
\begin{equation}
 T_{d'}(g)= \max_{\vec{v}_{x_1}\in\R^{d'},||\vec{v}_{x_1}||=1} \sum_{x_2=1}^{M_2}  \left|\left|\sum_{x_1=1}^{M_1} g_{x_1,x_2}
\vec{v}_{x_1}\right|\right|.
\end{equation}
As $T_{d'}(g)=T_{d'}(g^T)$, we can assume that $M_1\leq M_2$ without loss of generality. We give an example for such optimisations:

\begin{example}
We optimise inequality $D6_{1}$ in
Ref.~\cite{Gisin2007}. It is the skew left circulant matrix given by the first row $\left(\begin{array}{cccccc}
                                                                                           1 & 0 & 1 & 0 & 1 & 1
                                                                                          \end{array}\right)$, i.e.
\begin{equation}
 {D6}_1 = \left(
\begin{array}{cccccc}
 1 & 0 & 1 & 0 & 1 & 1 \\
 0 & 1 & 0 & 1 & 1 & -1 \\
 1 & 0 & 1 & 1 & -1 & 0 \\
 0 & 1 & 1 & -1 & 0 & -1 \\
 1 & 1 & -1 & 0 & -1 & 0 \\
 1 & -1 & 0 & -1 & 0 & -1 \\
\end{array}
\right)
\end{equation}
A solution $\alpha$ of Theorem~\ref{thm:tightness} is $\alpha=\sqrt{\frac{M}{d}}\1^d$, as it is the case for many circulant (left, right,
skew left, skew right) matrices. See \cite{Karner2003301} for the singular value decomposition of circulant matrices.
We applied modifications (\ref{mods:twist}) and (\ref{mods:shift}) of Corollary~\ref{cor:mods}. We
started with a global random search to find good starting points, which we further optimised by a local optimisation. Both algorithms are
numerical. This led us to the matrix
{\setlength{\mathindent}{0cm}
\begin{equation}
g'=\left(
\begin{array}{cccccc}
 -0.350174 & 0.323788 & 0.344416 & -0.368076 & -0.299221 & 0.31404 \\
 -0.472675 & -0.357842 & -0.182589 & -0.31764 & -0.377403 & 0.215713 \\
 -0.218507 & -0.300642 & -0.525576 & -0.185735 & 0.38952 & 0.279595 \\
 0.39405 & 0.286377 & -0.315566 & -0.315986 & 0.296399 & 0.391561 \\
 0.303896 & 0.37589 & -0.193803 & -0.514786 & -0.310722 & -0.200436 \\
 0.190791 & -0.355309 & -0.321679 & -0.184563 & -0.326631 & -0.511833 \\
\end{array}
\right),
\end{equation}
}
which corresponds to an inequality with a qutrit to qubit ratio of $T/T_3\approx 1.02622$.
This seems to be small. However, we do not know of a higher ratio than $1.03528$ with few settings (see $\mathcal{B}_{X4}$ in
\cite{Vertesi2009}, with $8+4$ settings).
\end{example}

\section{Conclusions}\label{sec:conclusions}
We presented two modifications of the coefficients of bipartite CHSH-type Bell inequalities, which preserve tightness of the Tsirelson bound
$T$ given in \cite{Epping2013}. Physically, they do not affect the optimal observables (up to a relative rotation of the two laboratories).
We applied this method to show that for any relative rotation of the two laboratories, there is a Bell inequality that is maximally violated
for this rotation and a fixed shared quantum state. Furthermore we optimised Bell inequalities with respect to the ratio of the quantum
value and the local hidden
varible bound. Finally we showed how our method can be used to optimise dimension witnessing Bell inequalities, i.e. Bell inequalities,
where the maximal quantum value is not achievable with two qubits.
\ack
We thank N. Gisin for discussions on \cite{Gisin1999}. This project was supported by the DFG, BMBF and SFF of Heinrich-Heine-University
Düsseldorf.
\section*{References}
\bibliographystyle{unsrt}
\bibliography{citations}

\end{document}